\documentclass[a4paper,10pt]{article}
\usepackage[utf8]{inputenc}
\usepackage{amsmath,amssymb,amsthm}
\usepackage{mathtools}
\usepackage{mymacros}
\usepackage{geometry}
\usepackage{shellesc}
\newcommand{\LL}{\mathbb{L}}
\newcommand{\HH}{\mathbb{H}}

\newtheorem{theorem}{Theorem}
\newtheorem{lem}{Lemma}
\newtheorem{remark}{Remark}
\newtheorem{example}{Example}
\usepackage{url}
\newcommand{\T}{\mathsf{T}}
\renewcommand{\tilde}{\widetilde}
\renewcommand{\hat}{\widehat}

\usepackage{todonotes}

\newcommand{\epl}{$\mathrm{e}+$}
\newcommand{\emi}{$\mathrm{e}-$}

\usepackage{multirow}
\usepackage{xspace}
\usepackage{pgfplots}
\usepackage{subfigure}
\usepackage{enumerate}
\usepackage[a4paper,breaklinks,colorlinks]{hyperref}
\hypersetup{
     linkcolor = red,
     citecolor = green,
     urlcolor = blue,
  }
\usepackage{tikz}
\usetikzlibrary{external}
\newcommand{\HSV}{Hankel singular values\xspace}

\title{Balanced Truncation Model Reduction with A Priori Error Bounds for LTI Systems with Nonzero Initial Value\thanks{This research has been supported by the Einstein Foundation Berlin within the framework of the Einstein Center for Mathematics (ECMath) in the project SE1: ``Reduced order modeling for data assimilation''. Most of this research has been carried out while the second author was at Universit\"at Hamburg and Technische Universit\"at Berlin. Their support is gratefully acknowledged.}}
\author{Christian Schr\"oder\thanks{Freelancing numerical analyst. Last academic position: Techni\-sche Universit\"at Berlin, Institut f\"ur Mathematik, Stra{\ss}e des 17. Juni 136, 10623 Berlin, Germany, E-Mail: \texttt{chris.schroeder@gmx.net}.} and Matthias Voigt\thanks{Corresponding author, UniDistance Suisse, Schinerstrasse 18, 3900 Brig, Switzerland, E-Mail: \texttt{matthias.voigt@fernuni.ch}.}}

\begin{document}

\maketitle

\begin{abstract}
 In standard balanced truncation model order reduction, the initial condition is typically ignored in the reduction procedure and is assumed to be zero instead. However, such a reduced-order model  may be a bad approximation to the full-order system, if the initial condition is not zero. In the literature there are several attempts for modified reduction methods at the price of having no error bound or only a posteriori error bounds which are often too expensive to evaluate. In this work we propose a new balancing procedure that is based on a shift transformation on the state. We first derive a joint projection reduced-order model in which the part of the system depending only on the input and the one depending only on the initial value are reduced at once and we prove an a priori error bound. With this result at hand, we derive a separate projection procedure in which the two parts are reduced separately. This gives the freedom to choose different reduction orders for the different subsystems. Moreover, we discuss how the reduced-order models can be constructed in practice. Since the error bounds are parameter-dependent we show how they can be optimized efficiently. We conclude this paper by comparing our results with the ones from the literature by a series of numerical experiments.
\end{abstract}

\section{Introduction}
We consider model order reduction (MOR) of asymptotically stable linear time invariant (LTI) dynamical systems of the form
\begin{align}\label{eq:LTI}
\begin{split}
 \dot{x}(t) &= Ax(t)+Bu(t), \quad x(0) = x_0, \\ 
       y(t) &= Cx(t)+Du(t) 
\end{split}
\end{align}
for $t\ge 0$ with initial value $x_0 = X_0z_0$. We assume that $A \in \R^{n \times n}$, $B \in\R^{n \times m}$, $C \in \R^{p \times n}$, and $D \in \R^{p \times m}$, where $m,\,p \ll n$. The functions $u:[0,\infty) \to \R^m$, $x:[0,\infty) \to \R^n$, and $u:[0,\infty) \to \R^p$ are the \emph{input}, \emph{state}, and \emph{output}, respectively. In this paper we assume that all feasible initial conditions live in a low-dimensional subspace of $\R^n$ spanned by the columns of the matrix $X_0 \in \R^{n \times q}$.

In this work we aim for a reduced order model (ROM) of state dimension $r \ll n$ of the form
\begin{subequations}\label{eq:ROM}
\begin{align}
 \dot{x}_r(t) &= A_r x_r(t) + B_r u(t), \quad x_r(0) = X_{0,r} z_0, \label{eq:ROMx} \\
       y_r(t) &= C_r x_r(t) + D_r u(t) + f(t), \label{eq:ROMy}
\end{align}
\end{subequations}
such that $y_r(\cdot)$ is close to $y(\cdot)$ for all inputs $u(\cdot)$ and all initial values defined by $z_0$.

We will present two methods to achieve this goal. The first method determines a ROM of state-space dimension $r$ with error bound
\begin{equation}\label{eq:bound_share}
 \left\|y-y_r\right\|_{\LL_2} \le 2(\eta_{r+1} + \cdots + \eta_n)(\left\|u\right\|_{\LL_2}+\beta\left\|z_0\right\|_2).
\end{equation}
The other method obtains a ROM of state-space dimension $r = k + \ell$ with error bound
\begin{equation}\label{eq:bound_sep}
 \left\|y-y_r\right\|_{\LL_2}\le 2(\sigma_{k+1} + \cdots + \sigma_n)\left\|u\right\|_{\LL_2} + 2(\theta_{\ell+1}+\cdots+\theta_n)\left\|z_0\right\|_2.
\end{equation}
Here, $\eta_i$, $\sigma_i$, and $\theta_i$, $i = 1,\,\ldots,\,n$ are the Hankel singular values of certain systems.

In this paper we will first prove these bounds. Then we give a practical procedure on how to construct the ROMs and evaluate it numerically. Moreover, we give a detailed comparison to other approaches. 

\section{State of the Art} \label{sec:sota}
In this section we review the current state of the art on balanced truncation model reduction. In particular, we will discuss several approaches for treating inhomogeneous initial conditions and give the error bounds, if available. It is important to note that the error bounds listed below are typically \emph{a posteriori bounds}, so the error bound is only available after the ROM has been computed. In constrast to that, our bounds \eqref{eq:bound_share} and \eqref{eq:bound_sep} are \emph{a priori bounds}, so they can be evaluated before the ROM is known. 

\paragraph{Balanced Truncation (BT)} is a well known method for model reduction of asymptotically stable LTI systems, see, e.\,g.,~\cite{GugA04}.
From given $A,\,B,\,C$ and a desired reduced order $r$, this method computes projection matrices $V_r,\,W_r \in \R^{n \times r}$ and Hankel singular values $\sigma_1,\,\ldots,\,\sigma_n$ which are independent of $r$.
We will use the notation
\begin{equation*}
 [V_r,W_r,\sigma_1,\ldots,\sigma_n] = \operatorname{BT}(A,B,C,r).
\end{equation*}
The ROM (of order $r$) is then given by projection, that is
\begin{subequations}\label{eq:proj}
\begin{align}
 A_r &:= W_r^\T A V_r, \quad B_r := W_r^\T B, \quad C_r := CV_r, \quad D_r := D,\label{eq:proj_abcd}\\
 X_{0,r} &:= W_r^\T X_0, \quad f \equiv 0. \label{eq:proj_x0f}
\end{align}
\end{subequations}

It is well-known that the ROM is asymptotically stable if $\sigma_r > \sigma_{r+1}$~\cite[Section~7.2]{Ant05}. We define $\LL_2^k := \left\{ g:[0,\infty) \to \R^k : \int_0^\infty g(t)^\T g(t){\rm d}t < \infty \right\}$ with induced norm $\left\|g\right\|_{\LL_2}:=\left(\int_0^\infty g(t)^\T g(t){\rm d}t\right)^{1/2}$. Analogously, we define $\LL_2^{k \times \ell} := \left\{ G:[0,\infty) \to \R^{k \times \ell} : \int_0^\infty G(t)^\T G(t){\rm d}t < \infty \right\}$ with induced norm $\left\|G\right\|_{\LL_2}:=\left(\int_0^\infty G(t)^\T G(t){\rm d}t\right)^{1/2}$. If $u \in \LL_2^m$, then due to asymptotic stability of the original system and the ROM we have $y,\,y_r \in \LL_2^p$. If further $x_0 = 0$, then BT guarantees the a priori error bound
\begin{equation}\label{eq:bt_bound}
 \left\|y-y_r\right\|_{\LL_2}\le 2(\sigma_{r+1}+\cdots+\sigma_n)\left\|u\right\|_{\LL_2}.
\end{equation}

This error bound can be extended to the case of inhomogeneous initial values as follows. We can write the outputs of the FOM and the ROM as
\begin{align*}
 y(t) &= C\mathrm{e}^{At}X_0z_0 + \int_{0}^t C\mathrm{e}^{A(t-\tau)}Bu(\tau) \mathrm{d}\tau + Du(t), \\
 y_r(t) &= C_r\mathrm{e}^{A_rt}X_{0,r}z_0 + \int_{0}^t C_r\mathrm{e}^{A_r(t-\tau)}B_ru(\tau) \mathrm{d}\tau + Du(t).
\end{align*}
Since in BT, both $A$ and $A_r$ are asymptotically stable, $C\mathrm{e}^{A\cdot}X_0,\, C_r\mathrm{e}^{A_r\cdot}X_{0,r} \in \LL_2^{p \times q}$ as well as $C\mathrm{e}^{A\cdot}X_0z_0,\, C_r\mathrm{e}^{A_r\cdot}X_{0,r}z_0 \in \LL_2^{p}$ and hence we obtain the estimate
\begin{align*}
 \left\|y-y_r\right\|_{\LL_2} &\le \left\| C\mathrm{e}^{A\cdot}X_0z_0 - C_r\mathrm{e}^{A_r\cdot}X_{0,r}z_0 \right\|_{\LL_2} + 2(\sigma_{r+1}+\cdots+\sigma_n)\left\|u\right\|_{\LL_2} \\ &
 \le \left\| C\mathrm{e}^{A\cdot}X_0 - C_r\mathrm{e}^{A_r\cdot}X_{0,r} \right\|_{\LL_2} {\| z_0 \|}_2 + 2(\sigma_{r+1}+\cdots+\sigma_n)\left\|u\right\|_{\LL_2}.
\end{align*}
However, $\left\| C\mathrm{e}^{A\cdot}X_0 - C_r\mathrm{e}^{A_r\cdot}X_{0,r} \right\|_{\LL_2}$ is nothing but the $\HH_2$-norm of the transfer function of the system $\left[ \begin{bsmallmatrix} A & 0 \\ 0 & A_r \end{bsmallmatrix}, \begin{bsmallmatrix} X_0 \\ X_{0,r} \end{bsmallmatrix}, \begin{bsmallmatrix} C & -C_r \end{bsmallmatrix} \right]$ and can be computed by solving a (typically large-scale) Lyapunov equation \cite[Chap. 4]{ZhoDG96}. The problem is that in standard BT, we do not attempt to minimize the part of the error that is related to the initial values. One can think of certain situations  in which this error is large. One possibility is the case in which $X_{0,r} = W_r^\T X_0 \approx 0$. Then the error associated with the initial values is essentially given by the $\HH_2$-norm of the transfer function of the system $[A,X_0,C]$ which can be large.


\paragraph{The method TrlBT of Baur, Benner, and Feng~\cite{BauBF14}} consists of translating the state $x(t)$ to $\tilde{x}(t) := x(t) - x_0$. The original system becomes
\begin{align*}
  \dot{\tilde{x}}(t) &= A\tilde{x}(t) + \begin{bmatrix}B & Ax_0 \end{bmatrix}\tilde{u}(t), \quad \tilde{x}(0)=0,\\
                y(t) &= C\tilde{x}(t) + \begin{bmatrix}D & Cx_0 \end{bmatrix}\tilde{u}(t) \quad \text{with }\tilde{u}(t):=\begin{bmatrix}u(t)\\1\end{bmatrix}.
\end{align*}
This homogeneous system is then reduced by standard balanced truncation. Note that BT is applied to an expanded system, here $\left[A,\begin{bmatrix}B & Ax_0\end{bmatrix},C\right]$. Since $\tilde{u} \notin \LL_2^{m+1}$, no computable error bound similarly to~\eqref{eq:bt_bound} can be given.

\paragraph{The method AugBT of Heinkenschloss, Reis, and Antoulas~\cite{HeiRA11}}
consists of applying BT to the expanded system $\left[A,\begin{bmatrix}B & X_0\end{bmatrix},C\right]$ to obtain the projection matrices $V_r,\,W_r$ and Hankel singular values $\eta_1,\,\ldots,\,\eta_n$. The ROM is obtained by~\eqref{eq:ROM} and~\eqref{eq:proj} using the modified projection matrices $V_r,\,W_r$. This method achieves an a posteriori error bound
\begin{multline} \label{eq:bound_HeiRA11}
 \left\|y-y_r\right\|_{\LL_2} \le 2(\eta_{r+1} + \cdots + \eta_n)\left\|u\right\|_{\LL_2}\\
+3\cdot 2^{-1/3}(\eta_{r+1}+\cdots+\eta_n)^{2/3}\left(\left\|L^\T AX_0\right\|_2+{\|\Sigma_r^{\frac12}A_rX_{0,r}\|}_2\right)^{1/3}\left\|z_0\right\|_2,
\end{multline}
where $L$ is such that $L L^\T$ is the observability Gramian of the expanded system, and $\Sigma_r = {\rm diag}(\eta_1,\ldots,\eta_r)$. This bound has two disadvantages: it involves the reduced system and the \HSV appear with exponent $2/3$.

The previously discussed methods are joint--projection methods, i.\,e., they use a single projection to produce ROMs in which both the responses to the input and initial state are treated simultaneously. On the other hand, there are seperate--projection methods producing ROMs in which the two parts are reduced seperately. This leads to a ROM that consists of two decoupled subsystems as in the following method. 

\paragraph{The method BT-BT of Beattie, Gugercin, and Mehrmann~\cite{BeaGM17}} produces a seperate--projection ROM.
Let
\begin{align*}
 [V_{k},W_{k},\sigma_1,\ldots,\sigma_n] &= \operatorname{BT}(A,B,C,k), \\
 \big[\hat{V}_{\ell},\hat{W}_{\ell},\theta_1,\ldots,\theta_n\big] &= \operatorname{BT}(A,X_0,C,\ell).
\end{align*}
Then a reduced order model of order $r = k+\ell$ is constructed as in~\eqref{eq:ROM} with
\begin{align*}
 A_r &:= \begin{bmatrix} W_{k}^\T A V_{k} & 0 \\ 0 & \hat{W}_{\ell}^\T A \hat{V}_{\ell} \end{bmatrix}, \quad B_r = \begin{bmatrix} W_{k}^\T B \\ 0 \end{bmatrix}, \quad C_r = \begin{bmatrix} CV_{k} & C\hat{V}_{\ell} \end{bmatrix}, \\
 D_r &:= D, \quad X_{0,r} := \begin{bmatrix} 0 \\ \hat{W}_{\ell}^\T X_0 \end{bmatrix}, \quad f \equiv 0.
\end{align*}
Let $\big[A_{\rm b},{X}_{0,\rm b},C_{\rm b}\big]$ be a fully balanced realization of $[A,X_0,C]$ and assume that $Y$ solves the Sylvester equation
\begin{equation*}
 A_{\rm b}^\T Y + Y\begin{bmatrix}I_{\ell} & 0\end{bmatrix} A_{\rm b} \begin{bmatrix}I_{\ell} \\ 0\end{bmatrix} + C_{\rm b}^\T C_{\rm b}\begin{bmatrix}I_{\ell} \\ 0\end{bmatrix} =0.
\end{equation*}
With 
\begin{equation*}
 T := [t_{i,j}] = X_{0,\rm b}X_{0,\rm b}^\T + 2 Y \begin{bmatrix} I_\ell & 0 \end{bmatrix} A_{\rm b},
\end{equation*}
this method achieves an a posteriori error bound
\begin{equation} \label{eq:bound_BeaGM17}
 \left\|y-y_r\right\|_{\LL_2} \le 2(\sigma_{k+1} + \cdots + \sigma_n)\left\|u\right\|_{\LL_2}+ \left(t_{\ell+1,\ell+1}\theta_{\ell+1} + \cdots + t_{n,n}\theta_n\right)^{1/2}\left\|z_0\right\|_2.
\end{equation}

This bound has several disadvantages: even though, typically the values of $t_{i,i}$ for $i=\ell+1,\,\ldots,\,n$ are small, the \HSV appear with an exponent, here $1/2$.
Moreover, a fully balanced realization of $[A,X_0,C]$ is necessary, whose computation is expensive and can be numerically unstable. Also, the matrix $T$ depends on the reduced order $\ell$, making deciding on $\ell$ difficult a priori.

\paragraph{Singular perturbation approximations} are another class of model reduction techniques that are somewhat related to balanced truncation. Recently, the paper by Daragmeh, Hartmann, and Qatanani~\cite{DarHQ19} suggest a singular perturbation approximation for systems with nonzero initial condition. There, the authors provide another a posteriori error bound that is in the flavor of the one of~\cite{BeaGM17} and which can be computed by solving a Sylvester equation that contains data of the reduced-order model.

\section{Proposed Method}
In this section we discuss the main contribution of our paper. Here we derive two kinds of ROMs. The first one is a joint projection ROM in which a system with expanded system matrix is reduced by BT and in which both the system responses to the input and the initial condition are reduced at once. This ROM admits the error bound~\eqref{eq:bound_share}. Thereafter, we discuss a separate projection ROM in which both responses are reduced individually which leads to the error bound~\eqref{eq:bound_sep}. Since both ROMs are depending on design parameters we will have a detailed discussion about their interpretation and choice and we will discuss how to construct the ROMs in practice.

\subsection{Joint Projection ROMs}\label{sec:method}
Our first method consists of applying BT to a system with expanded input matrix. More precisely, the projection matrices $V_r,\,W_r$ and the {\HSV} $\eta_1,\,\ldots,\,\eta_n$ are obtained by
\begin{equation}\label{eq:BT_exp_joint}
[V_r,W_r,\eta_1,\ldots,\eta_n] = \operatorname{BT}\left(A, \begin{bmatrix}B & \frac1{\beta\sqrt{2\alpha}}(A+\alpha I_n)X_0\end{bmatrix},C,r\right),
\end{equation}
where $\alpha$ and $\beta$ are real positive parameters. 

The ROM is then given by~\eqref{eq:ROM}, where the reduced matrices $A_r$, $B_r$, $C_r$, and $D_r$ are given by~\eqref{eq:proj_abcd}, and $X_{0,r}$ and $f$ (in contrast to~\eqref{eq:proj_x0f}) are given by
\begin{subequations}\label{eq:Xfr}
\begin{align}
 X_{0,r} &:= (A_r+\alpha I_r)^{-1}W_r^\T(A+\alpha I_n)X_0   \label{eq:x0r}\\
 f(t) &:= F_rz_0{\rm e}^{-\alpha t}\quad\text{with } F_r:=CX_0-C_rX_{0,r}. \label{eq:fr}
\end{align}
\end{subequations}
We call this reduction method \emph{joint--projection decaying shift balanced truncation}, or shortly \textbf{jShiftBT}, because the derivation of this method involves a decaying shift of the system state (cf. the proof of the following theorem).

For this ROM we will prove the above mentioned error bound in the following theorem.
\begin{theorem}\label{thm:bound}
Let the LTI system~\eqref{eq:LTI} be asymptotically stable. Let $\alpha,\,\beta$ be real positive scalars. If $\eta_r > \eta_{r+1}$, then the ROM~\eqref{eq:ROM} defined by~\eqref{eq:BT_exp_joint},~\eqref{eq:proj_abcd}, and~\eqref{eq:Xfr} is asymptotically stable and the error is bounded by~\eqref{eq:bound_share}, provided that $-\alpha$ is not an eigenvalue of $A_r$. Moreover, we have $y_r(0) = y(0)$.
\end{theorem}
\begin{proof}
The proof consists of the derivation of the ROM and proceeds in three steps.

\emph{Step 1:} First, similarly to the method of Baur, Benner, and Feng~\cite{BauBF14}, the state is shifted to $\tilde{x}(t):=x(t)-x_0{\rm e}^{-\alpha t}$, i.\,e., the shift decays with rate $\alpha$.
Then $\tilde{x}(0)=0$ and for $\dot{\tilde{x}}(\cdot)$ we obtain
\begin{align*}
\dot{\tilde{x}}(t) & =\dot{x}(t)+\alpha x_0{\rm e}^{-\alpha t} \\
&= Ax(t)+Bu(t)+\alpha x_0{\rm e}^{-\alpha t}\\
&= A\big(\tilde{x}(t)+x_0{\rm e}^{-\alpha t}\big)+Bu(t)+\alpha x_0{\rm e}^{-\alpha t} \\
&= A\tilde{x}(t)+Bu(t)+(A+ \alpha I_n)x_0{\rm e}^{-\alpha t}.
\end{align*}
For the output we get
\begin{equation*}
 y(t) = C\big(\tilde{x}(t)+x_0{\rm e}^{-\alpha t}\big)+Du(t)=C\tilde{x}(t)+Du(t)+\frac{\beta\sqrt{2\alpha}}{\beta\sqrt{2\alpha}}CX_0z_0{\rm e}^{-\alpha t}.
\end{equation*}
Thus we obtain a linear system with homogeneous initial condition, i.\,e.,
\begin{align*} 
 \dot{\tilde{x}}(t) &= A\tilde{x}(t) + \begin{bmatrix} B & \frac1{\beta\sqrt{2\alpha}}(A+\alpha I_n)X_0 \end{bmatrix} \tilde{u}(t), \quad \tilde{x}(0)=0, \\
 y(t) &= C\tilde{x}(t) + \begin{bmatrix} D & \frac1{\beta\sqrt{2\alpha}}CX_0 \end{bmatrix} \tilde{u}(t) \quad \text{with }\tilde{u}(t) = \begin{bmatrix}u(t) \\ \beta\sqrt{2\alpha}{\rm e}^{-\alpha t}z_0 \end{bmatrix}.
\end{align*}

\emph{Step 2:} We apply standard BT to this system, which amounts to~\eqref{eq:BT_exp_joint} and obtain the ROM
\begin{subequations}
\begin{align}
 \dot{\tilde{x}}_r(t) &= A_r \tilde{x}_r(t) + \begin{bmatrix} B_r & \frac1{\beta\sqrt{2\alpha}}W_r^\T(A+\alpha I_n)X_0 \end{bmatrix} \tilde{u}(t), \quad \tilde{x}_r(0) = 0  \label{eq:ROMproofx} \\
     y_r(t) &= C_r\tilde{x}_r(t)+ \begin{bmatrix}D & \frac1{\beta\sqrt{2\alpha}}CX_0 \end{bmatrix} \tilde{u}(t)\label{eq:ROMproofy}
\end{align}
as in~\eqref{eq:proj}.
\end{subequations}
For $t=0$ we have
\begin{equation*}
 y_r(0) = C_r\tilde{x}_r(0) + \begin{bmatrix}D & \frac1{\beta\sqrt{2\alpha}}CX_0 \end{bmatrix} \tilde{u}(0)
=Du(0)+CX_0z_0=y(0).
\end{equation*}
Since $\eta_r > \eta_{r+1}$ by assumption, the ROM is asymptotically stable~\cite[Section~7.2]{Ant05}.
If $u \in \LL_2^m$, then we have $\tilde{u} \in \LL_2^{m+q}$ and therefore, $y,\,y_r \in \LL^p$ and by~\eqref{eq:bt_bound} we have 
\begin{equation}\label{eq:boundtilde}
 \left\|y-y_r\right\|_{\LL_2}\le 2(\eta_{r+1}+\cdots +\eta_n)\left\|\tilde{u}\right\|_{\LL_2}.
\end{equation}
Inserting
\begin{equation*}
\left\| \tilde{u} \right\|_{\LL_2} \le \left\| u \right\|_{\LL_2} + \beta \big\|\sqrt{2\alpha}{\rm e}^{-\alpha \cdot }z_0\big\|_{\LL_2} =\left\|u\right\|_{\LL_2} + \beta \left\|z_0\right\|_2 \big\|\sqrt{2\alpha}{\rm e}^{-\alpha \cdot} \big\|_{\LL_2}
=\left\|u\right\|_{\LL_2} + \beta \left\|z_0\right\|_2
\end{equation*}
into~\eqref{eq:boundtilde} yields the claimed bound~\eqref{eq:bound_share}.

\emph{Step 3:} We set $x_r(t) := \tilde{x}_r(t) + x_r(0){\rm e}^{-\alpha t}$ (i.\,e., we ``unshift'' $\tilde{x}_r(\cdot)$).
Here we set $x_r(0) = X_{0,r}z_0$ for an $X_{0,r}$ that is yet to be determined. Putting this $x_r(\cdot)$ into~\eqref{eq:ROMproofx} yields
\begin{align*}
 \dot{x}_r(t) &= A_rx_r(t) + B_ru(t) + \left(W_r^\T(A+\alpha I_n)X_0-(A_r+\alpha I_r\big)X_{0,r}\right)z_0{\rm e}^{-\alpha t},
\end{align*}
which reduces to~\eqref{eq:ROMx} by choosing $X_{0,r}$ as in~\eqref{eq:x0r}.
Inserting $x_r(\cdot)$ into~\eqref{eq:ROMproofy} yields
\begin{align*}
 y_r(t) &= C_r x_r(t) + \begin{bmatrix} D & \frac1{\beta\sqrt{2\alpha}}CX_0 \end{bmatrix} \tilde{u}(t) -C_rX_{0,r}z_0{\rm e}^{-\alpha t} \\
        &= C_rx_r(t) + Du(t) + F_rz_0{\rm e}^{-\alpha t},
\end{align*}
which is~\eqref{eq:ROMy} for $F_r$ as in~\eqref{eq:fr}.
This concludes the proof. 
\end{proof}

\begin{remark}
The property $y_r(0)=y(0)$ is shared with \textbf{TrlBT}.
The other methods do not guarantee this.
\end{remark}
\begin{remark}
The term $f(t) = F_rz_0{\rm e}^{-\alpha t}$ in our ROM is non-standard.
It consists of a vector times a scalar-valued exponential function. Hence it is easy to compute.
However, if a ROM without $f$ is desired we can get rid of it at the price of expanding the ROM.
Indeed, the ROM may be reformulated by appending $\phi(t):={\rm e}^{-\alpha t}$ to $x_r(t)$. This gives
\begin{align*}
\begin{bmatrix}\dot{x}_r(t)\\ \dot{\phi}(t)\end{bmatrix}
&=\begin{bmatrix}A_r&0\\ 0&-\alpha\end{bmatrix}
\begin{bmatrix}x_r(t)\\ \phi(t)\end{bmatrix}
+\begin{bmatrix}B_r\\ 0\end{bmatrix}u, \quad
\begin{bmatrix}x_r(0)\\ \phi(0)\end{bmatrix}=
\begin{bmatrix}X_{0,r}z_0\\ 1\end{bmatrix}\\
y_r &= \begin{bmatrix}C_r & F_rz_0 \end{bmatrix}
\begin{bmatrix}x_r(t) \\ \phi(t)\end{bmatrix}
+D_ru(t).
\end{align*}

However, the initial value is not linear in $z_0$ anymore (but affine linear).
Also, the output matrix of the ROM depends on $z_0$, which may be undesirable.
These disadvantages can be removed by reformulating the ROM by appending $\psi(t):=R_rz_0{\rm e}^{-\alpha t}$ to $x_r(t)$. This yields
\begin{align*}
\begin{bmatrix}\dot{x}_r(t)\\ \dot{\psi}(t)\end{bmatrix}
&=\begin{bmatrix}A_r&0\\ 0&-\alpha I\end{bmatrix}
\begin{bmatrix}x_r(t) \\ \psi(t)\end{bmatrix}
+\begin{bmatrix}B_r\\ 0\end{bmatrix}u(t), \quad
\begin{bmatrix}x_r(0)\\ \psi(0)\end{bmatrix}=
\begin{bmatrix}X_{0,r}\\ R_r\end{bmatrix}z_0\\
y_r(t) &= \begin{bmatrix} C_r & L_r \end{bmatrix}
\begin{bmatrix}x_r(t) \\ \psi(t) \end{bmatrix}
+D_ru(t),
\end{align*}
where $F_r =: L_rR_r$ is a rank-revealing decomposition of $F_r$. This ROM is completely in standard form, but is of order $r+ \operatorname{rank}(F_r)$ which is bounded from above by $r + \min\{p,q\}$.
\end{remark}
\begin{remark}
As a case study consider the case $q=1$ and $X_0$ being an eigenvector of $A$ corresponding to an eigenvalue $-\alpha \in \R$.
Then $(A+\alpha I)X_0=0$ and our method gives the same projection matrices as standard BT.
Now if $X_0$ happens to be both, an uncontrollable mode and easy to observe, it will be truncated in the ROM, i.\,e., $W_r^\T X_0=0$ and $X_{0,r}=0$.
On the other hand it has, as an initial value, significant influence on the output $y(\cdot)$.
How can our method work in this situation when standard BT does not? The answer lies in the extra term $f(\cdot)$ that reduces to $CX_0e^{-\alpha t}$ in this case, i.\,e., it reintroduces the mode that has been truncated by BT.
\end{remark}

\subsection{Separate Projection ROMs}
\label{sec:composite}
The ROM constructed in Subsection~\ref{sec:method} is a joint projection ROM 
where one has to specify the parameter $\beta$ to put an emphasis either on the input or on the initial condition.
However, the reduction error may be large, if e.\,g., a high weight is put on the input error ($\beta$ is large) and if $\left\| z_0 \right\|_2$ is large, since then the expression $\left\|u\right\|_{\LL_2}+\beta\left\|z_0\right\|_2$ is large, too. 

So the motivation of this subsection is to reduce the two parts of the system individually and to construct a seperate--projection ROM out of the two reduced subsystems similarly as in~\cite{BeaGM17}.

To begin with, we write the output of the system~\eqref{eq:LTI} as
\begin{equation*}
 y(t) = \underbrace{C\mathrm{e}^{At}x_0}_{=: y_{x_0}(t)} + \underbrace{\int_{0}^t C\mathrm{e}^{A(t-\tau)}Bu(\tau) \mathrm{d}\tau + Du(t)}_{=:y_u(t)}.
\end{equation*}
The output component $y_u(\cdot)$ is given by the output of the system
\begin{align} \label{eq:LTIhom}
 \begin{split}
   \dot{x}(t) &= Ax(t) + Bu(t), \quad x(0) = 0, \\
       y_u(t) &= Cx(t) + Du(t),
 \end{split}     
\end{align}
while $y_{x_0}(\cdot)$ is the output of the system
\begin{align} \label{eq:LTIinh}
 \begin{split}
   \dot{\hx}(t) &= A \hx(t), \quad \hx(0) = x_0 = X_0z_0, \\
           y_{x_0}(t) &= C \hx(t).
 \end{split}
\end{align}
Now we reduce the system~\eqref{eq:LTIhom} using standard balanced truncation, i.\,e.,
\begin{equation*}
 [V_{k},W_{k},\sigma_1,\,\ldots,\sigma_n] = \operatorname{BT}(A,B,C,k)
 \end{equation*}
 and the ROM is given by 
 \begin{align*}
  \dot{x}_k(t) &= A_k x_k(t) + B_k u(t), \quad x_k(0) = 0, \\
        y_{u,k}(t) &= C_k x_k(t) + D_k u(t)
 \end{align*}
 with $A_k = W_k^\T A V_k$, $B_k = W_k^\T B$, $C_k = CV_k$, and $D_k = D$.

The system~\eqref{eq:LTIinh} is reduced using the approach from Subsection~\ref{sec:method} for $\beta = 1$. This results in performing balanced trunction on a shifted system, i.\,e.,
 \begin{equation}\label{eq:BT_exp_sep}
 \big[\hat{V}_{\ell},\hat{W}_{\ell},\theta_1,\,\ldots,\theta_n\big] = \operatorname{BT}\left(A,\frac{1}{\sqrt{2\alpha}}(A+\alpha I_n)X_0,C,\ell\right)
\end{equation}
and the corresponding ROM is 
\begin{align*}
 \dot{\hx}_{\ell}(t) &= \hat{A}_{\ell} \hx_\ell(t), \quad \hx_\ell(0) = X_{0,\ell} z_0, \\ 
         y_{x_0,\ell}(t) &= \hat{C}_{\ell} \hx_\ell(t) + \hat{F}_\ell z_0 \mathrm{e}^{-\alpha t},
\end{align*}
with $\hat{A}_\ell = \hat{W}_\ell^\T A \hat{V}_\ell$, $\hat{C}_\ell = C \hat{V}_\ell$, $X_{0,\ell} = \big(\hA_\ell+\alpha I_\ell\big)^{-1}\hat{W}_\ell^\T(A+\alpha I_n)X_0$, and $\hat{F}_\ell = CX_0 - \hat{C}_\ell X_{0,\ell}$.

With the reduced subsystems above we can now construct the overall ROM
\begin{align}\label{eq:compositeROM}
\begin{split}
 \begin{bmatrix} \dot{x}_k(t) \\ \dot{\hx}_\ell(t) \end{bmatrix} &= \begin{bmatrix} A_k & 0 \\ 0 & \hat{A}_\ell \end{bmatrix} \begin{bmatrix} x_k(t) \\ \hat{x}_\ell(t) \end{bmatrix} + \begin{bmatrix} B_k \\ 0 \end{bmatrix} u(t), \quad \begin{bmatrix} x_k(0) \\ \hat{x}_\ell(0) \end{bmatrix} = \begin{bmatrix} 0 \\ X_{0,\ell} z_0 \end{bmatrix}, \\
 y_r(t) &:= y_{u,k}(t) + y_{x_0,\ell}(t) = \begin{bmatrix} C_k & \hat{C}_\ell \end{bmatrix}  \begin{bmatrix} x_k(t) \\ \hat{x}_\ell(t) \end{bmatrix} + D_k u(t) + \hat{F}_\ell z_0 \mathrm{e}^{-\alpha t}.
 \end{split}
\end{align}
We call this method \emph{separate--projection decaying shift balanced truncation}, shortly \textbf{sShiftBT}.

Now the following result is an immediate consequence of a combination of the standard BT error bound and Theorem~\ref{thm:bound}.
\begin{theorem}
Let the LTI system~\eqref{eq:LTI} be asymptotically stable. Let $\alpha$ be a real positive scalar. If $\sigma_k > \sigma_{k+1}$ and $\theta_\ell > \theta_{\ell+1}$, then the ROM~\eqref{eq:compositeROM} is asymptotically stable and the error is bounded by~\eqref{eq:bound_sep}, provided that $-\alpha$ is not an eigenvalue of $\hat{A}_\ell$. Moreover, we have $y_r(0) = y(0)$.
\end{theorem}

\begin{remark} An advantage of the separate projection ROM is that it is a feasible ROM for all possible inputs and initial values. In contrast to that one may have to construct several joint projection ROMs for different values of $\beta$ in order to cover all possible inputs and initial values one wants to simulate the model with.

Moreover, note that the reduced order of the separate projection ROM is $r = k+\ell$. However, since the reduced state matrix is of block-diagonal structure, the two reduced subsystems can be simulated individually. In particular, we have $z_0 = \sum_{i=1}^q \zeta_i x_{0}^{(i)}$ with a basis $\left\{ x_{0}^{(1)},\,\ldots,\,x_{0}^{(q)} \right\}$ of $\operatorname{im} X_{0,\ell}$. Thus, if the model has to be simulated for a lot of different initial conditions, one could precompute
\begin{equation*}
 {y}_{x_0,\ell}^{(i)}(t) := \hat{C}_\ell \mathrm{e}^{\hat{A}_\ell t} x_{0}^{(i)}, \quad i = 1,\,\ldots,\,q.
\end{equation*}
Then for the particular initial condition $x_0 = X_0z_0$,
\begin{equation*}
 {y}_{x_0,\ell}(t) = \sum_{i=1}^q \zeta_i {y}_{x_0,\ell}^{(i)}(t)  
\end{equation*}
can be evaluated more efficiently and the online costs are dominated by a ROM of  reduced order $k$.
\end{remark}

\subsection{Discussion of the Parameters $\alpha$ and $\beta$}\label{sec:disc}
All that remains is to choose the two parameters $\alpha$ and $\beta$ in the method.
Let us begin by noting that for $\alpha\rightarrow 0$ our decaying--shift approach $\tilde{x}(t)=x(t)-x_0{\rm e}^{-\alpha t}$ reduces to the constant one of \textbf{TrlBT}. 
Also, for $\alpha \rightarrow \infty$, the function $\left(\frac{1}{\sqrt{2\alpha}}{\rm e}^{-\alpha \cdot}\right)^2$ converges to Dirac's $\delta$ impulse used in \textbf{AugBT} and \textbf{BT-BT}. 
Moreover, for $\beta \to \infty$ we obtain the standard BT ROM.
In that sense our approach contains the existing ones.

However, for all these extreme cases of $\alpha$ or $\beta$ approaching zero or $\infty$, our error bound gets huge! This is either, because $\beta$ appears explicitly in it, or because the expanded input matrix $B_{\rm exp} := \begin{bmatrix} B & \frac{1}{\beta\sqrt{2\alpha}}(A+\alpha I_n)X_0\end{bmatrix}$ contains large elements leading to large \HSV.
So, good values for $\alpha$ and $\beta$ are neither too large nor too small.

With $c_u:=2(\eta_{r+1}+\cdots+\eta_n)$ and $c_{x_0}:= \beta c_u$ we write the error bound~\eqref{eq:bound_share} as
\begin{equation}\label{eq:bound_with_cu} 
 \left\| y-y_r \right\|_{\LL_2} \le c_u\cdot\left\|u\right\|_{\LL_2} + c_{x_0} \cdot {\|z_0\|}_2.
\end{equation}
Note that all, $c_u$, $c_{x_0}$, and the \HSV $\eta_i$ depend on $\alpha$ and $\beta$.
We will also write $c_u(\alpha,\beta)$ and $c_{x_0}(\alpha,\beta)$ whenever we want to emphasize this dependency.

Observe that both summands of~\eqref{eq:bound_with_cu} are influenced by $\alpha$ in the same way. Hence $\alpha$ is a tuning parameter, i.\,e., optimizing it for $c_u$ also  improves the value of $c_{x_0}$.
An ad hoc heuristic candidate is the choice
\begin{equation*}
 \alpha_{{\rm heur}}=\frac{{\|AX_0\|}_{\rm F}}{{\|X_0\|}_{\rm F}}
\end{equation*}
which minimizes $\big\|\tfrac{1}{\sqrt{\alpha}}(A+\alpha I_n)X_0\big\|_{\rm F}$, i.\,e., the norm of the extra block in $B_{\rm exp}$. 
Another possibility that certainly comes to mind is the negative spectral abscissa
\begin{equation*}
 \tilde{\alpha}_{{\rm heur}} = -\max_{\lambda\in \Lambda(A)}\operatorname{Re}(\lambda).
\end{equation*}
With this choice, the shift $x(t)-\tilde{x}(t) = x_0{\rm e}^{-\alpha t}$ decays at the same rate as the most stable mode of the homogeneous system $\dot{x}(t) = A x(t)$.
Of course, $\alpha$ can also be obtained by numerical optimization methods, see Subsection~\ref{sec:alpha} for details.
We will assess these aproaches with numerical examples in Section~\ref{sec:eval}.

For $\beta$, things are different because it influences $c_u$ and $c_{x_0}$ in different ways. Note that by our construction of $B_{\rm exp}$, $c_u$ is monotonically decreasing in $\beta$ whereas $c_{x_0} = \beta c_u$ is monotonically increasing.
Thus, by increasing $\beta$ we improve the input part of the error bound $c_u\cdot{\|u\|}_{\LL_2}$, but worsen the initial value part $c_{x_0}\cdot {\|z_0\|}_2$, and vice versa.
So, if nothing is known about $u$ and $z_0$, \emph{then $\beta$ should be considered a design parameter that is provided by the user}.
For example, if we want $c_u$ to be a hundred times smaller than $c_{x_0}$, then we set $\beta=100$.

However, in certain situations, $\beta$ can also be a tuning parameter.
Assuming that (typical or approximate) values of ${\|u\|}_{\LL_2}$ and ${\|z_0\|}_2$, or at least their ratio, are known, we can optimize the right hand side of~\eqref{eq:bound_with_cu} over $\alpha$ and $\beta$. In this situation an ad hoc value would be
\begin{equation}\label{eq:betaadhoc}
 \beta_{\rm heur}={{\|u\|}_{\LL_2}}/{{\|z_0\|}_2},
\end{equation}
as this choice balances the two summands in~\eqref{eq:bound_with_cu}.
It turns out that this choice is almost optimal with suboptimality factor at most 2 as the following lemma shows. Thus, further numerical optimization for $\beta$ is rather futile.
\begin{lem}
Let ${\|u\|}_{\LL_2}$ and ${\|z_0\|}_2$ be given and choose $\beta_{\rm heur}={{\|u\|}_{\LL_2}}/{{\|z_0\|}_2}$. For $\alpha,\,\beta > 0$ define $$e(\alpha,\beta) := c_u(\alpha,\beta) (\left\|u\right\|_{\LL_2} + \beta {\|z_0\|}_2).$$
Then for any fixed $\alpha_0 > 0$ we have
\begin{equation} \label{eq:subopt1}
 e(\alpha_0,\beta_{\rm heur}) \le 2 \min_{\beta>0} e(\alpha_0,\beta).
\end{equation}
Moreover,
\begin{equation} \label{eq:subopt2}
 \min_{\alpha >0} e(\alpha,\beta_{\rm heur}) \le 2 \min_{\alpha,\beta>0} e(\alpha,\beta).
\end{equation}
%
\end{lem}
\begin{proof}
First we show \eqref{eq:subopt1}: First, for $\alpha,\,\beta >0$ define $$\mu({\alpha},\beta):=\max\left\{c_u(\alpha,\beta)\left\|u\right\|_{\LL_2},\, c_{x_0}(\alpha,\beta) {\|z_0\|}_2\right\}.$$
Now let $\alpha_0>0$ be arbitrary but fixed. First note that $\beta_{\rm heur}$ minimizes $\mu({\alpha_0},\cdot)$ because of the monotonicity of $c_u(\alpha_0,\cdot)$ and $c_{x_0}(\alpha_0,\cdot)$ and since
\begin{equation*}
c_u(\alpha_0,\beta_{\rm heur})\left\|u\right\|_{\LL_2} = c_{x_0}(\alpha_0,\beta_{\rm heur}) {\|z_0\|}_2 = \beta_{\rm heur} c_{u}(\alpha_0,\beta_{\rm heur}) {\|z_0\|}_2.
\end{equation*}
Then for all $\beta > 0$ we obtain
\begin{equation*}
  e(\alpha_0,\beta_{\rm heur})
 = 2\mu(\alpha_0,\beta_{\rm heur})
 \le 2\mu(\alpha_0,\beta)
 \le 2e(\alpha_0,\beta).
\end{equation*}

Next we show \eqref{eq:subopt2}: Define $\alpha_* := \argmin_{\alpha > 0} e(\alpha,\beta_{\rm heur})$ and $\big(\hat{\alpha}_*,\hat{\beta}_*\big) = \argmin_{\alpha,\beta > 0} e(\alpha,\beta)$. Then with the help of \eqref{eq:subopt1} we obtain the estimate
\begin{equation*}
   e(\alpha_*,\beta_{\rm heur}) \le e\big(\hat{\alpha}_*,\beta_{\rm heur}\big) \le 2 e\big(\hat{\alpha}_*,\hat{\beta}_*\big),
\end{equation*}
which concludes the proof.
\end{proof}


\subsection{Efficient Construction of the ROMs}
If $\alpha$ and $\beta$ are known, then we can construct the ROMs by~\eqref{eq:BT_exp_joint} or~\eqref{eq:BT_exp_sep} using a BT implementation of choice. However, if the ROMs have to be determined for several choices of $\alpha$ and $\beta$, or if $\alpha$ and $\beta$ are to be determined inside an optimization loop, then this can get prohibitively expensive.
In this subsection we explain, how the reduced-order models presented in Subsections~\ref{sec:method} and~\ref{sec:composite} can be efficiently determined for many values of $\alpha$ and $\beta$.

Consider the three Lyapunov equations
\begin{subequations}\label{eq:lyap}
\begin{align}
 A P + P A^\T + B B^\T = 0, \label{eq:lyapB} \\
 A \hP + \hP A^\T + X_0 X_0^\T = 0, \label{eq:lyapX0} \\
 A^\T Q + Q A + C^\T C = 0. \label{eq:lyapC} 
\end{align}
\end{subequations}
Assume that we have computed low-rank factorizations of the solutions
\begin{equation*}
 P = R R^\T, \quad \hP = \hR \hR^\T, \quad Q = L L^\T. 
\end{equation*}
These factors are directly obtained by many established methods such as the ADI method or Krylov subspace methods (see, e.\,g.,~\cite{BenS13,Sim07}) for which well-tested software exists, e.\,g.,~\cite{BenKS20}. Recall that in Subsection~\ref{sec:method} we want to reduce the system $\left[A,\begin{bmatrix}B & \frac{1}{\beta\sqrt{\alpha}}(A + \alpha I_n) X_0 \end{bmatrix}, C \right]$, so we need to determine its controllability Gramian given by the solution of the Lyapunov equation
\begin{equation*}
 A \cP(\alpha,\beta) + \cP(\alpha,\beta) A^\T + \begin{bmatrix}B & \frac{1}{\beta\sqrt{2\alpha}}(A + \alpha I_n) X_0 \end{bmatrix}\begin{bmatrix}B & \frac{1}{\beta\sqrt{2\alpha}}(A + \alpha I_n) X_0 \end{bmatrix}^\T = 0.
\end{equation*}
By multiplying~\eqref{eq:lyapX0} by $\frac{1}{\beta\sqrt{2\alpha}}(A + \alpha I_n)$ from the left and by $\frac{1}{\beta\sqrt{2\alpha}}(A + \alpha I_n)^\T$ from the right and adding~\eqref{eq:lyapB}, we see that
\begin{equation*}
 \cP(\alpha,\beta) = P + \left(\frac{1}{\beta\sqrt{2\alpha}}(A + \alpha I_n)\right) \hat{P} \left(\frac{1}{\beta\sqrt{2\alpha}}(A + \alpha I_n)\right)^\T.
\end{equation*}
In particular, we have the factorization 
\begin{equation*}
 \cP(\alpha,\beta) = \cR(\alpha,\beta)\cR(\alpha,\beta)^\T = \begin{bmatrix} R & \frac{1}{\beta\sqrt{2\alpha}}(A + \alpha I_n) \hat{R} \end{bmatrix}\begin{bmatrix} R & \frac{1}{\beta\sqrt{2\alpha}}(A + \alpha I_n) \hat{R} \end{bmatrix}^\T.
\end{equation*}
Now the ROM is determined using the SVD of the matrix
\begin{equation}\label{eq:M}
  M(\alpha,\beta) := L^\T \cR(\alpha,\beta) = \begin{bmatrix} L^\T R & \frac{1}{\beta\sqrt{2\alpha}}L^\T A \hat{R} + \frac{\sqrt{\alpha}}{\sqrt{2}\beta} L^\T \hat{R}  \end{bmatrix}.
\end{equation}
In particular, the nonzero \HSV $\eta_i$ are the nonzero singular values of $M(\alpha, \beta)$.
Note that $L^\T R$, $L^\T A \hat{R}$, and $L^\T \hat{R}$ are independent of $\alpha$ and $\beta$ and are typically all small matrices that can be efficiently precomputed.
Note that solving the three Lyapunov equations is by far the dominant computational burden.
Computing the SVDs, even for many values of $\alpha$ and $\beta$ is comparatively cheap.
This allows computing the ROM inside an optimization loop to obtain optimal values of $\alpha$ and $\beta$. This will be done in the following section.
\begin{remark}
The construction of the separate projection ROM from Subsection~\ref{sec:composite} and \mbox{\textbf{BT-BT}} also makes use of the same three Lyapunov equations in~\eqref{eq:lyap}.
Also the ROMS of \textbf{TrlBT} and \textbf{AugBT}, while constructed differently in~\cite{BauBF14} and~\cite{HeiRA11}, could be built by solving these three equations.
\end{remark}

\subsection{Optimizing the Parameter $\alpha$}\label{sec:alpha} 
Now we return to the optimization of $\alpha$.
The value $\alpha_{\rm heur}$ is no more than an educated guess; with the results of the previous subsection we are ready to use numerical optimization machinery.
First consider the bound~\eqref{eq:bound_share}, namely
\begin{equation*}
  \left\| y - y_{r}\right\|_{\LL_2} \le \underbrace{2 \left( \eta_{r+1}(\alpha) + \ldots + \eta_n(\alpha) \right)}_{=:
  c_u(\alpha)}\left( \left\| u \right\|_{\LL_2} + \beta \left\| z_0 \right\|_2\right).
\end{equation*}
Our goal is to find $\alpha_*$ such that $c_{u}(\alpha_*)$ is minimal as this will minimize the error bound. Note that the Hankel singular values $\eta_i(\alpha)$ and hence $y_r(\cdot)$ also depend on $\beta$. However, the value of $\beta$ as well as the reduced order $r$ are fixed, so we do not list them explicitly as arguments since we only focus on the optimization of $\alpha$ in this subsection. One should however be aware of the fact that the optimization has to be repeated for every reduced order and parameter $\beta$ of interest, because the optimal value $\alpha_*$ depends on both.
  
First, recall that the nonzero Hankel singular values $\eta_i(\alpha)$ are the nonzero singular values of the matrix $M(\alpha) = \begin{bmatrix} L^\T R & \frac{1}{\beta\sqrt{2\alpha}}L^\T A \hat{R} + \frac{\sqrt{\alpha}}{\sqrt{2}\beta} L^\T \hat{R} \end{bmatrix}$.
Hence, $c_{u}(\cdot)$ is continous and piecewise smooth.
The only critical points $\alpha$, where $c_u(\cdot)$ may not be smooth, are those for which $\eta_{r}(\alpha)$ and $\eta_{r+1}(\alpha)$ coincide or where the smallest nonzero Hankel singular value goes through zero.
The latter is usually several orders of magnitude smaller than $\eta_{r+1}(\cdot)$ and hence affects $c_u(\cdot)$ only marginally. Therefore, we do not consider this case any further here. However, note that in principal, the problem of minimizing $c_u(\cdot)$ is a \emph{non-smooth problem}, meaning that local minima may be attained at points, at which $c_u(\cdot)$ is not differentiable.

Next we show that, under a weak assumption, the case $\eta_r(\alpha_0) = \eta_{r+1}(\alpha_0)$ cannot lead to a local minimum at $\alpha_0$ as summarized in the following lemma.
\begin{lem}
 Assume that $\eta_{r+1}(\cdot)$ is not differentiable at $\alpha_0 > 0$ and that $\eta_{r-1}(\alpha_0) > \eta_r(\alpha_0) = \eta_{r+1}(\alpha_0) > \eta_{r+2}(\alpha_0)$.
 Let $M(\cdot)$ have constant rank in a neighborhood of $\alpha_0$. 
 Then $\alpha_0$ is not a local minimizer of $c_u(\cdot)$.
\end{lem}
\begin{proof}
Since the singular value curves of $M(\cdot)$ can be chosen to be real analytic, the function $\eta_{r+1}(\cdot)$ is \emph{semi-differentiable} at $\alpha_0$ and the left and right derivatives 
\begin{align*}
 \frac{\mathrm{d}_-}{\mathrm{d} \alpha} \eta_{r+1}(\alpha_0) := \lim_{\alpha \to \alpha_0^-} \frac{\eta_{r+1}(\alpha) - \eta_{r+1}(\alpha_0)}{\alpha - \alpha_0},  \quad\frac{\mathrm{d}_+}{\mathrm{d}\alpha} \eta_{r+1}(\alpha_0) := \lim_{\alpha \to \alpha_0^+} \frac{\eta_{r+1}(\alpha) - \eta_{r+1}(\alpha_0)}{\alpha - \alpha_0}
\end{align*}
exist. We further have 
\begin{enumerate}[a)]
 \item $\frac{\mathrm{d}_+}{\mathrm{d}\alpha}\eta_{r+1}(\alpha_0)=\frac{\mathrm{d}_-}{\mathrm{d}\alpha}\eta_{r}(\alpha_0)$, since ${\eta_r(\cdot)|}_{(\alpha_0-\varepsilon, \alpha_0)}$ and ${\eta_{r+1}(\cdot)|}_{[\alpha_0,\alpha_0+\varepsilon)}$ form a smooth singular value curve for some $\varepsilon > 0$; 
 \item $\frac{\mathrm{d}_-}{\mathrm{d}\alpha}\eta_{r}(\alpha_0) <\frac{\mathrm{d}_-}{\mathrm{d}\alpha}\eta_{r+1}(\alpha_0)$, since $\eta_{r+1}(\cdot)$ is the smaller singular value;
 \item the function $c_u(\alpha)-\eta_{r+1}(\alpha)=\sum_{i=r+2}^n \eta_i(\alpha)$ is smooth in a neighborhood of $\alpha_0$ (using $\eta_{r+1}(\alpha_0) > \eta_{r+2}(\alpha_0)$ and the constant--rank assumption).
\end{enumerate}
Together we have
\begin{equation*}
 \frac{\mathrm{d}_+}{\mathrm{d}\alpha}c_u(\alpha_0)-\frac{\mathrm{d}_-}{\mathrm{d}\alpha}c_u(\alpha_0)
 \stackrel{\text{c)}}{=}\frac{\mathrm{d}_+}{\mathrm{d}\alpha}\eta_{r+1}(\alpha_0)-\frac{\mathrm{d}_-}{\mathrm{d}\alpha}\eta_{r+1}(\alpha_0)
 \stackrel{\text{a)}}{=}\frac{\mathrm{d}_-}{\mathrm{d}\alpha}\eta_{r}(\alpha_0)-\frac{\mathrm{d}_-}{\mathrm{d}\alpha}\eta_{r+1}(\alpha_0)
 \stackrel{\text{b)}} < 0,
\end{equation*}
i.\,e., the right derivative of $c_u(\cdot)$ jumps downwards at $\alpha_0$. However, for $\alpha_0$ to be a local minimum, the right derivative would have to jump from a negative to a positive value.
Thus, $\alpha_0$ is not a local minimum.
\end{proof}


Figure~\ref{fig:alphaplot} shows a typical plot of $c_{u}(\cdot)$.
\begin{figure}
 \centering
 \includegraphics{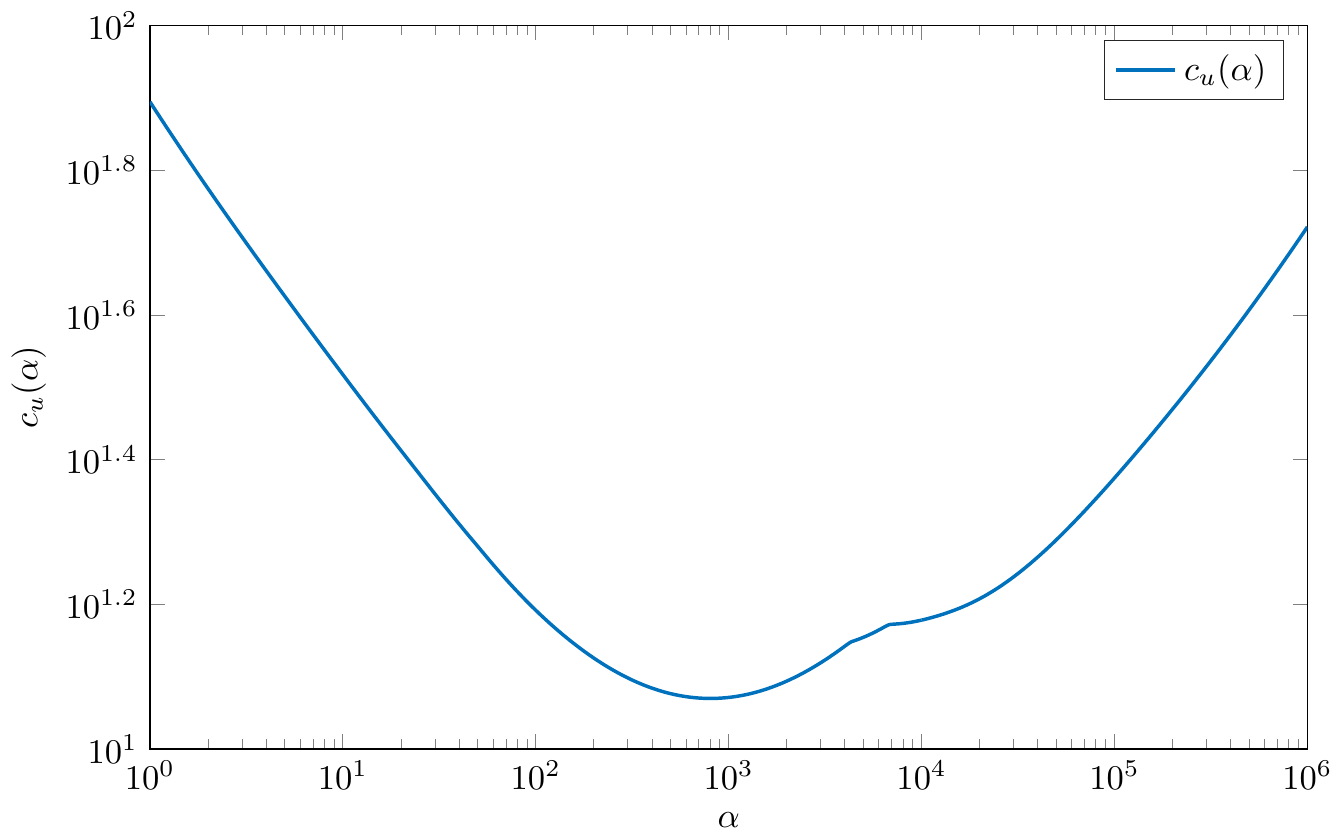}
 \caption{Typical behavior of $c_u(\cdot)$. Note that there are two non-differentiable points near $10^4$.}
 \label{fig:alphaplot}
\end{figure}
We observe three segments: for small values of $\alpha$, $c_{u}(\cdot)$ is monotonically decreasing, thereafter follows a region, in which $c_u(\cdot)$ is non-monotonic and contains local minima and after that, for large values of $\alpha$, $c_u(\cdot)$ is monotonically increasing.
This can be explained by looking at~\eqref{eq:M}:
for small (or large) $\alpha$, the term $\tfrac{1}{\beta\sqrt{2\alpha}}L^\T A\hat{R}$ (or $\tfrac{\sqrt{\alpha}}{\beta\sqrt{2}}L^\T A\hat{R}$) is dominating $M(\alpha)$ and its singular values -- in contrast to the case of a medium-sized $\alpha$ for which the three terms in $M(\alpha)$ are of comparable size. 

Aiming for a local minimum of $c_u(\cdot)$ at least, we proceed in the following steps:

\emph{Step 1: We use a sampling procedure} to obtain a good value for $\alpha$. One possibility is to sample $c_u(\alpha)$ for a large range of magnitudes in order to find at least approximately the region in which local minima are present. One possibility is to choose $\alpha \in \{10^j \; | \; j \in j_{\min}, j_{\min}+1,\, \ldots,\,j_{\max}\}$ with integers $j_{\min} \le j_{\max}$.


\emph{Step 2: We perform an optional local optimization} to improve the best sample value $c_u(\alpha)$ from Step 1.
In our approach we use standard gradient-based optimization methods implemented in the \textsc{Matlab} function \texttt{fmincon}. To address the possible non-smooth nature of the problem, one could also use more sophisticated non-smooth optimization methods such as \texttt{GRANSO} \cite{CurMO17}.
It remains to determine the derivative of $c_{u}(\cdot)$.
As discussed above, these do in general not exist for all $\alpha$.
Nevertheless, the derivatives exist almost everywhere.
We need the following result for the derivative of the singular values.
\begin{lem}[\cite{Lan64}] \label{lem:diffsingval}
Consider the differentiable function $Z : (-\varepsilon,\varepsilon) \to \R^{n \times m}$. Let $\sigma(t)$ be a singular value of $Z(t)$ converging to a simple singular value $\sigma_0$ of $Z_0 := Z(0)$. Let $u_0 \in \R^n$ and $v_0 \in \R^m$ be the corresponding right and left singular vectors, i.\,e., $\left\| u_0 \right\|_2 = \left\| v_0 \right\|_2 = 1$, $Z_0 u_0 = \sigma_0 v_0$ and $v_0^\T Z_0 = \sigma_0 u_0^\T$. Then
\begin{equation*}
 \left.\frac{\mathrm{d}}{\mathrm{d}t}\sigma(t)\right|_{t=0} = v_0^\T \left(\left.\frac{\mathrm{d}}{\mathrm{d}t}Z(t)\right|_{t=0}\right) u_0. 
\end{equation*}
\end{lem}
We have
\begin{align*}
 \frac{\mathrm{d}}{\mathrm{d} \alpha} M(\alpha) &= \begin{bmatrix} 0 & -\frac{1}{2\beta\alpha\sqrt{2\alpha}} L^\T A \hR + \frac{1}{2\beta\sqrt{2\alpha}} L^\T \hR  \end{bmatrix}
\end{align*}
and with Lemma~\ref{lem:diffsingval} we finally obtain
\begin{align*}
 \frac{\mathrm{d}}{\mathrm{d} \alpha} c_{u}(\alpha) =
 2\sum_{j=r+1}^n\frac{\mathrm{d}}{\mathrm{d} \alpha} \eta_j(\alpha) 
 &= 2\sum_{j=r+1}^n v_j(\alpha)^\T \frac{\mathrm{d}}{\mathrm{d} \alpha} M(\alpha) u_j(\alpha), 
\end{align*}
where $u_j(\alpha)$ and $v_j(\alpha)$ are the right and left singular vectors associated with the singular value $\eta_j(\alpha)$ of $M(\alpha)$ which we assume to be simple in order to ensure differentiability.

The second bound~\eqref{eq:bound_sep} 
\begin{equation*}
 \left\| y - y_r\right\|_{\LL_2} \le 2 \left( \sigma_{r+1} + \ldots + \sigma_n \right)\left\|u \right\|_{\LL_2} + 2 \left( \theta_{r+1}(\alpha) + \ldots + \theta_n(\alpha) \right)\left\| z_0 \right\|_{2}
\end{equation*}
can be treated in a similar manner as the bound above. Note that the first summand is parameter-independent and the second one only depends on $\alpha$ in the singular values of
\begin{equation*}
  N(\alpha) = \frac{1}{\sqrt{2\alpha}}L^\T A \hR + \frac{\sqrt{\alpha}}{\sqrt{2}} L^\T \hR.
 \end{equation*}
The derivatives can be obtained as for $M(\alpha)$ by setting $\beta = 1$ and therefore, we omit the details here.

\section{Numerical Evaluation}\label{sec:eval}
Now we evaluate our method and especially the error bounds and compare them with the methods listed in Section~\ref{sec:sota}. Here we consider two small examples from the SLICOT benchmark collection for model reduction\footnote{See \url{http://slicot.org/20-site/126-benchmark-examples-for-model-reduction}.}, see~\cite{ChaV02}. In principle, our method also works well in the large-scale setting since methods for solving large-scale Lyapunov equations are available and the optimization procedure from Subsection~\ref{sec:alpha} only acts on matrices which are constructed by low-rank Cholesky factors which are usually small. However, we choose to consider small examples since they allow us to evaluate the error of the reduction by simulating the error systems. The examples are constructed such that the input is zero initially so that we we can assess the quality of the reduction for the part of the ROM depending on the initial values. Later, the input is turned on and we can evaluate the reduction of the response to the input.
\begin{example} \label{exm:beam} First we consider the \texttt{beam} example with $n=348$ and $m=p=1$. We choose $X_0 = \big[X_0^{(i,j)}\big] \in \R^{n \times q}$ with $q=2$, $X_0^{(5,1)} = 1$, $X_0^{(101,2)} = 100$, and zeros elsewhere. As input we choose $$u(t) = \begin{cases} 1, & \text{if } t \in [500,1000], \\ 0, &  \text{otherwise}\end{cases}$$ with ${\|u\|}_{\LL_2} = 500$ and as initial condition we take $x_0 = X_0z_0$ with $z_0 = \left[\begin{smallmatrix} 10 \\ -1 \end{smallmatrix}\right]$ and ${\|z_0\|}_2 \approx 10.0499$.
\end{example}

\begin{example} \label{exm:CDplayer} The second example we consider in this paper is the \texttt{CDplayer} with $n=120$ and $m=p=2$. The $X_0 \in \R^{n \times q}$ with $q=2$ is constructed such that $W_r^\T X_0 = 0$, where $W_r$ is the left projection matrix obtained from standard \textbf{BT} with $r=50$. In this way we aim to construct an example where \textbf{BT} leads to a poor reduction in the part of the response that depends on the initial value. The input is chosen as $$u(t) = \begin{bmatrix} 0 \\ 1 \end{bmatrix} \cdot \begin{cases} 1, & \text{if } t \in [1.5,3], \\ 0, &  \text{otherwise}\end{cases}$$ with ${\|u \|}_{\LL_2} = 1.5$. The initial condition is $x_0 = X_0z_0$ with $z_0 = \left[\begin{smallmatrix} 1 \\ 10 \end{smallmatrix}\right]$ and ${\|z_0\|}_2 \approx 10.0499$.
\end{example}
The following numerical experiments have been run under \textsc{Matlab} R2021b Update 1 on a HP X360 Convertible laptop with an Intel\textsuperscript{\textregistered} Core\textsuperscript{\texttrademark} i7-10710U CPU @ 1.10 GHz with 16 GB of RAM and using Windows 10.

\subsection{Evaluation of the Error Bounds}
First we consider the error bound constants $c_u$ and $c_{x_0}$ as in~\eqref{eq:bound_with_cu} of \textbf{jShiftBT} for several fixed values of $\beta$ and compare them with the ones obtained by \textbf{AugBT} and \textbf{BT}.
We have used the optimized values $\alpha_*$ for the error bounds.
The results are listed in Table~\ref{tab:joint}.
As discussed in Subsection~\ref{sec:disc}, the table illustrates that for \textbf{jShiftBT}, $c_u$ is monotonically decreasing and that it tends to the value of $c_u$ for \textbf{BT} for $\beta \to \infty$. Moreover, $c_{x_0}$ for \textbf{jShiftBT} is monotonically increasing. 

Let us briefly discuss the performance of the heuristic choices $\alpha_{\rm heur}$ and $\tilde{\alpha}_{\rm heur}$. A comparison is given in Table~\ref{tab:alpha} in which we list the values of the different $\alpha$ values and the corresponding value of $c_u$ for \textbf{jShiftBT} with $\beta = 1$ and $r=30$. Recall that in this case, $c_u = c_{x_0}$, so we only list one of the values. Apparantly, the choice of $\tilde{\alpha}_{\rm heur}$ is overestimating the best error bound obtained by $\alpha_*$ by up to two orders of magnitude. On the other hand, the choice $\alpha_{\rm heur}$ is closer to the optimal value and there is only a small overestimation of the error bound.  


Next, we evaluate the error bound~\eqref{eq:bound_sep} of the separate projection ROMs and compare it with the a posteriori bound~\eqref{eq:bound_BeaGM17} for \textbf{BT-BT}, respectively. For both methods, we list the values of $c_u$ and $c_{x_0}$ in Table~\ref{tab:allmethods} for various partial reduced order $k$ and $\ell$. Since $c_u$ is the same for both methods we only list it once. First, we see that for all reduced orders, the bounds for \textbf{BT-BT} are typically better than the ones for \textbf{sShiftBT}. However, we would like to stress that the a posteriori bound of \textbf{BT-BT} requires a fully balanced realization of the system $[A,X_0,C]$ which cannot be computed for large systems and moreover the solution of a Sylvester equation for every value of $\ell$.  On the other hand, our method works on the low-rank Gramian factors and these can still be optimized even in the large-scale setting. So our bound can still be at least estimated in practical applications.


\begin{table}
 \centering
 \caption{Comparison of the error bound constants $c_u/c_{x_0}$ for \textbf{AugBT}, \textbf{BT}, and \textbf{jShiftBT} for several fixed values of $\beta$. The parameter $\alpha$ has been optimized as in Subsection~\ref{sec:alpha}.}
 \label{tab:joint}
\subtable[Results for Example~\ref{exm:beam}.]{\begin{tabular}{c|c|ccccc|c}
    $r$  & \textbf{AugBT} & \multicolumn{5}{c|}{\textbf{jShiftBT}} & \textbf{BT} \\
       & & $\beta=0.01$ & $\beta=0.1$ & $\beta = 1$ & $\beta = 10$ & $\beta = 100$ & \\
  \hline 
  \multirow{2}{*} 5    & 2.4\epl 2/& 2.9\epl 4/& 3.4\epl 3/& 4.1\epl 2/& 1.7\epl 2/& 1.7\epl 2/& 1.7\epl 2/ \\
                       & 2.5\epl 2 & 2.9\epl 2 & 3.4\epl 2 & 4.1\epl 2 & 1.7\epl 3 & 1.7\epl 4 & 1.0\epl 1 \\
  \multirow{2}{*} {10} & 5.1\epl 1/& 1.2\epl 4/& 1.2\epl 3/& 1.3\epl 2/& 2.9\epl 1/& 2.4\epl 1/& 2.4\epl 1/ \\
                       & 9.4\epl 1 & 1.2\epl 2 & 1.2\epl 2 & 1.3\epl 2 & 2.9\epl 2 & 2.4\epl 3 & 8.4\epl 0 \\
  \multirow{2}{*} {15} & 2.1\epl 1/& 5.0\epl 3/& 5.0\epl 2/& 5.3\epl 1/& 1.1\epl 1/& 7.7\epl 0/& 7.5\epl 0/ \\
                       & 5.8\epl 1 & 5.0\epl 1 & 5.0\epl 1 & 5.3\epl 1 & 1.1\epl 2 & 7.7\epl 2 & 3.8\epl 0 \\
  \multirow{2}{*} {20} & 1.2\epl 1/& 2.8\epl 3/& 2.8\epl 2/& 3.1\epl 1/& 6.4\epl 0/& 3.8\epl 0/& 3.7\epl 0/ \\
                       & 4.1\epl 1 & 2.8\epl 1 & 2.8\epl 1 & 3.1\epl 1 & 6.4\epl 1 & 3.8\epl 2 & 2.9\epl 0 \\
  \multirow{2}{*} {25} & 6.7\epl 0/& 1.4\epl 3/& 1.4\epl 2/& 1.7\epl 1/& 3.7\epl 0/& 1.9\epl 0/& 1.8\epl 0/ \\
                       & 2.8\epl 1 & 1.4\epl 1 & 1.4\epl 1 & 1.7\epl 1 & 3.7\epl 1 & 1.9\epl 2 & 3.0\epl 0 \\
  \multirow{2}{*} {30} & 3.5\epl 0/& 5.8\epl 2/& 5.8\epl 1/& 7.4\epl 0/& 2.0\epl 0/& 9.3\emi 1/& 8.6\emi 1/ \\
                       & 1.9\epl 1 & 5.8\epl 0 & 5.8\epl 0 & 7.4\epl 0 & 2.0\epl 1 & 9.3\epl 1 & 1.3\epl 0 \\
  \multirow{2}{*} {40} & 1.2\epl 0/& 1.9\epl 2/& 1.9\epl 1/& 2.6\epl 0/& 7.4\emi 1/& 2.5\emi 1/& 2.0\emi 1/ \\
                       & 9.3\epl 0 & 1.9\epl 0 & 1.9\epl 0 & 2.6\epl 0 & 7.4\epl 0 & 2.5\epl 1 & 6.8\emi 1 \\
  \multirow{2}{*} {50} & 4.2\emi 1/& 4.7\epl 1/& 4.9\epl 0/& 8.4\emi 1/& 2.2\emi 1/& 6.2\emi 2/& 3.3\emi 2/ \\
                       & 4.5\epl 0 & 4.7\emi 1 & 4.9\emi 1 & 8.4\emi 1 & 2.2\epl 0 & 6.2\epl 0 & 3.8\emi 1 \\
  \end{tabular}}
  \subtable[Results for Example~\ref{exm:CDplayer}.]{\begin{tabular}{c|c|ccccc|c}
    $r$  & \textbf{AugBT} & \multicolumn{5}{c|}{\textbf{jShiftBT}} & \textbf{BT} \\
       & & $\beta=0.01$ & $\beta=0.1$ & $\beta = 1$ & $\beta = 10$ & $\beta = 100$ & \\
  \hline 
  \multirow{2}{*} 5    & 1.3\epl 3/& 1.0\epl 5/& 1.5\epl 4/& 2.8\epl 3/& 1.5\epl 3/& 1.3\epl 3/& 1.3\epl 3/ \\
                       & 1.8\epl 4 & 1.0\epl 3 & 1.5\epl 3 & 2.8\epl 3 & 1.5\epl 4 & 1.3\epl 5 & 3.3\epl 1 \\
  \multirow{2}{*} {10} & 7.0\epl 1/& 5.2\epl 4/& 7.1\epl 3/& 1.1\epl 3/& 2.1\epl 2/& 7.6\epl 1/& 6.3\epl 1/ \\
                       & 2.5\epl 3 & 5.2\epl 2 & 7.1\epl 2 & 1.1\epl 3 & 2.1\epl 3 & 7.6\epl 3 & 3.3\epl 1  \\
  \multirow{2}{*} {15} & 2.0\epl 1/& 3.7\epl 4/& 4.5\epl 3/& 5.4\epl 2/& 1.1\epl 2/& 2.5\epl 1/& 1.2\epl 1/ \\
                       & 1.1\epl 3 & 3.7\epl 2 & 4.5\epl 2 & 5.4\epl 2 & 1.1\epl 3 & 2.5\epl 3 & 3.3\epl 1  \\
  \multirow{2}{*} {20} & 1.2\epl 1/& 2.7\epl 4/& 3.2\epl 3/& 3.8\epl 2/& 5.7\epl 1/& 1.5\epl 1/& 4.7\epl 0/ \\
                       & 7.6\epl 2 & 2.7\epl 2 & 3.2\epl 2 & 3.8\epl 2 & 5.7\epl 2 & 1.5\epl 3 & 3.3\epl 1  \\
  \multirow{2}{*} {25} & 7.1\epl 0/& 2.1\epl 4/& 2.3\epl 3/& 2.8\epl 2/& 4.0\epl 1/& 8.6\epl 0/& 1.6\epl 0/ \\
                       & 6.7\epl 2 & 2.1\epl 2 & 2.3\epl 2 & 2.8\epl 2 & 4.0\epl 2 & 8.6\epl 2 & 3.3\epl 1  \\
  \multirow{2}{*} {30} & 4.2\epl 0/& 1.7\epl 4/& 1.8\epl 3/& 2.1\epl 2/& 3.0\epl 1/& 4.5\emi 0/& 8.1\emi 1/ \\
                       & 4.8\epl 2 & 1.7\epl 2 & 1.8\epl 2 & 2.1\epl 2 & 3.0\epl 2 & 4.5\epl 2 & 3.3\epl 1  \\
  \multirow{2}{*} {40} & 2.1\epl 0/& 1.1\epl 4/& 1.1\epl 3/& 1.2\epl 2/& 1.6\epl 1/& 2.3\emi 0/& 2.9\emi 1/ \\
                       & 3.0\epl 2 & 1.1\epl 2 & 1.1\epl 2 & 1.2\epl 2 & 1.6\epl 2 & 2.3\epl 2 & 3.3\epl 1  \\
  \multirow{2}{*} {50} & 9.1\emi 1/& 7.2\epl 3/& 7.4\epl 2/& 7.8\epl 1/& 8.6\epl 0/& 1.1\emi 0/& 1.1\emi 1/ \\
                       & 1.7\epl 2 & 7.2\epl 1 & 7.4\epl 1 & 7.8\epl 1 & 8.6\epl 1 & 1.1\epl 2 & 3.3\epl 1  \\
  \end{tabular}}
\end{table}

\begin{table}[t]
 \centering
  \caption{Comparison of the error bound constants of \textbf{jShiftBT} for the heuristic values $\alpha_{\rm heur}$, $\tilde{\alpha}_{\rm heur}$, and the locally optimal choice $\alpha_*$. Here we use $\beta = 1$ (thus, $c_u = c_{x_0}$) and $r = 30$.}
  \label{tab:alpha}
 \subtable[Results for Example~\ref{exm:beam}]{
\begin{tabular}{c|ccc}
        & $\alpha_{*}$ & $\alpha_{\rm heur}$ & $\tilde{\alpha}_{\rm heur}$ \\
        \hline 
  $\alpha$ value & 1.1\epl 1 & 1.4\epl 2 & 5.1\emi 3 \\
  $c_u$ & 7.4\epl 0 & 1.5\epl 1 & 1.8\epl 2 \\
  \end{tabular}} 
  \subtable[Results for Example~\ref{exm:CDplayer}]{
\begin{tabular}{c|ccc}
        & $\alpha_{*}$ & $\alpha_{\rm heur}$ & $\tilde{\alpha}_{\rm heur}$ \\
        \hline
  $\alpha$  value & 5.6\epl 3 & 3.8\epl 4 & 2.4\emi 2 \\
  $c_u$ & 2.1\epl 2 & 3.0\epl 2 & 2.9\epl 4 \\
  \end{tabular}} 
\end{table}

\begin{table}[t]
 \centering
 \caption{Error bound constants of \textbf{BT-BT} and \textbf{sShiftBT} for the partial reduced orders $k$ and $\ell$. Here we use optimized values of $\alpha$. The minimum constants $c_{x_0}$ are emphasized by bold font.}
 \label{tab:allmethods}
 \subtable[Results for Example~\ref{exm:beam}]{
\begin{tabular}{c|c|cc}
    $k,\,\ell$ & \multicolumn{1}{c}{$c_u$} & \multicolumn{2}{c}{$c_{x_0}$} \\
     &  & \textbf{BT-BT} & \textbf{sShiftBT} \\
  \hline 
   5 & 1.7e$+$2 & {\bf 1.0e$+$1} & 2.9e$+$2 \\
  10 & 2.4e$+$1 & {\bf 7.0e$+$0} & 1.2e$+$2 \\
  15 & 7.5e$+$0 & {\bf 2.8e$+$0} & 5.0e$+$1 \\
  20 & 3.7e$+$0 & {\bf 1.7e$+$0} & 2.8e$+$1 \\
  25 & 1.8e$+$0 & {\bf 1.6e$+$0} & 1.4e$+$1 \\
  30 & 8.6e$-$1 & {\bf 2.9e$-$1} & 5.8e$+$0 \\
  40 & 2.0e$-$1 & {\bf 1.0e$-$1} & 1.9e$+$0 \\
  50 & 3.3e$-$2 & {\bf 5.9e$-$2} & 4.7e$-$1
  \end{tabular}} 
  \hspace*{0.2cm}\subtable[Results for Example~\ref{exm:CDplayer}]{\begin{tabular}{c|c|cc}
    $k,\,\ell$  & \multicolumn{1}{c}{$c_u$} & \multicolumn{2}{c}{$c_{x_0}$} \\
     &  & \textbf{BT-BT} & \textbf{sShiftBT} \\
  \hline 
   5 & 1.3e$+$3 & {\bf 9.0e$+$0} & 6.8e$+$2 \\
  10 & 6.3e$+$1 & {\bf 5.7e$+$0} & 4.1e$+$2 \\
  15 & 1.2e$+$1 & {\bf 6.3e$+$0} & 2.9e$+$2 \\
  20 & 4.7e$+$1 & {\bf 4.1e$+$0} & 2.3e$+$2 \\
  25 & 1.6e$+$0 & {\bf 3.4e$+$0} & 1.9e$+$2 \\
  30 & 8.1e$-$1 & {\bf 3.4e$+$0} & 1.6e$+$2 \\
  40 & 2.9e$-$1 & {\bf 1.8e$+$0} & 1.1e$+$2 \\
  50 & 1.1e$-$1 & {\bf 1.5e$+$0} & 6.8e$+$1
  \end{tabular}}
\end{table}

\subsection{Evaluation of the Errors}
In Figure~\ref{fig:allmethods} we plot the actual errors for each of the reduction methods. Here we choose $r=30$ for standard \textbf{BT}, \textbf{TrlBT}, \textbf{AugBT}, and \textbf{jShiftBT} as well as $k = \ell = 15$ for \textbf{BT-BT} and \textbf{sShiftBT}. These figures have been created by simulating the error systems using \textsc{Matlab}'s ODE solver \texttt{ode45} and plotting ${\|y(t) - y_r(t)\|}_2$ over the time $t$. Note that the results typically show an extremely oscillatory behavior, thus we have applied a smoothing filter to the output to improve the visibility of the individual error trajectories. The examples have been constructed such that we see the error of the response to the initial value in the first half of the considered time interval, while in the second half we see the error of the input-dependent part. The figure indicates that for Example~\ref{exm:beam}, the overall reduction using our \textbf{jShiftBT} outperforms all the other methods, especially when reducing the response to the initial value. The same is true for Example~\ref{exm:CDplayer}, even though it is not as clearly visible in the figure as for the first example (cf. Table~\ref{tab:errors} which contains the errors for our experiments). Standard \textbf{BT}, while not guaranteeing any error bound, also works similarly well for the two examples, as does \textbf{AugBT}.
However, the two separate projection methods \textbf{BT-BT} and \textbf{sShiftBT} perform much worse. But this behavior can be expected since the other methods use a single projection to reduce both the responses to the initial value and the input at once. In constrast to that, in the seperate projection methods, two individual projections which may have a large redundant ``overlap'', have to be formed. Finally, the method \textbf{TrlBT} does not seem to work well on both examples and produces by far the largest error.


\begin{figure}
 \centering
 \begin{tabular}{c}
 \subfigure[Results for Example~\ref{exm:beam}.]{
 \includegraphics{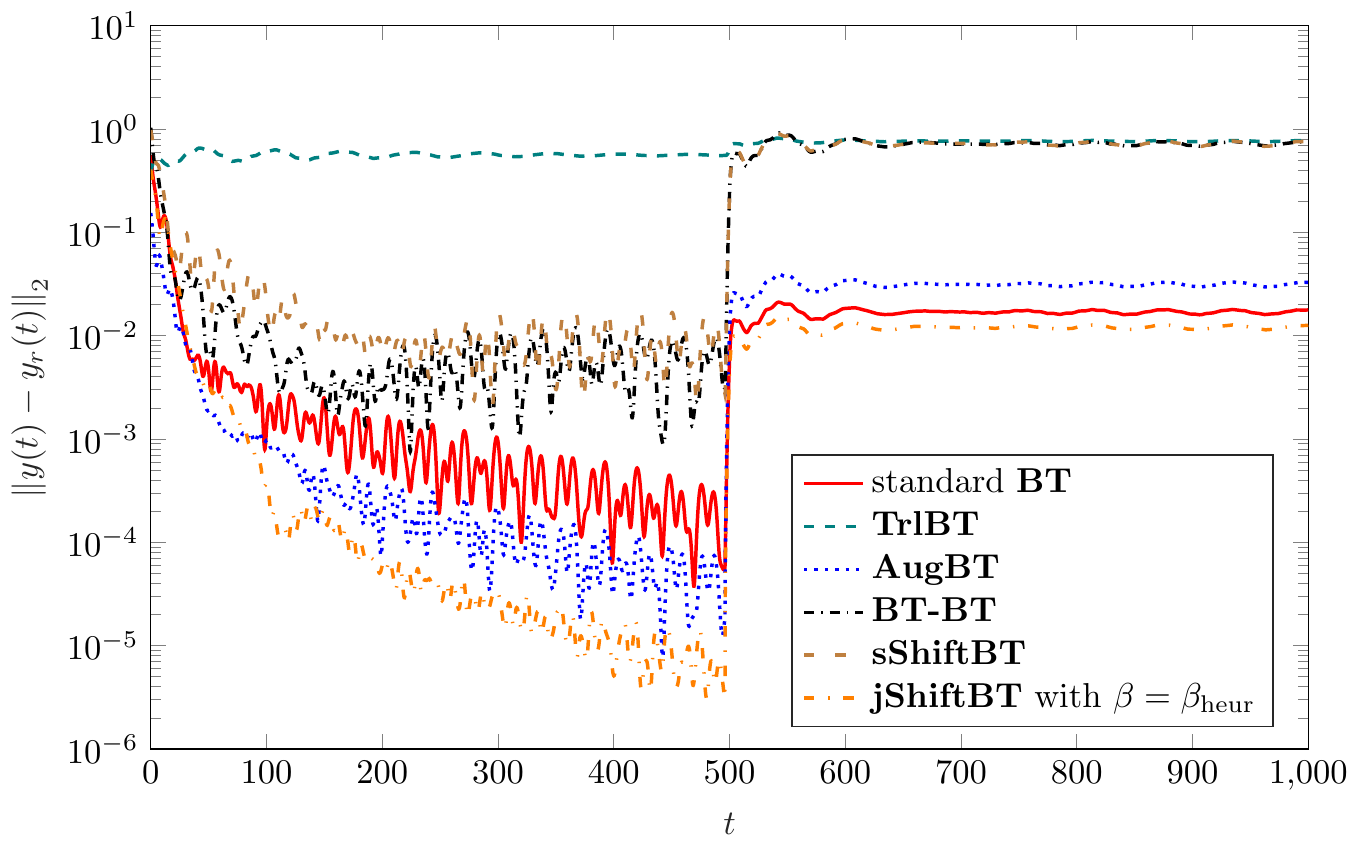}
 } \\
 \subfigure[Results for Example~\ref{exm:CDplayer}.]{
 \includegraphics{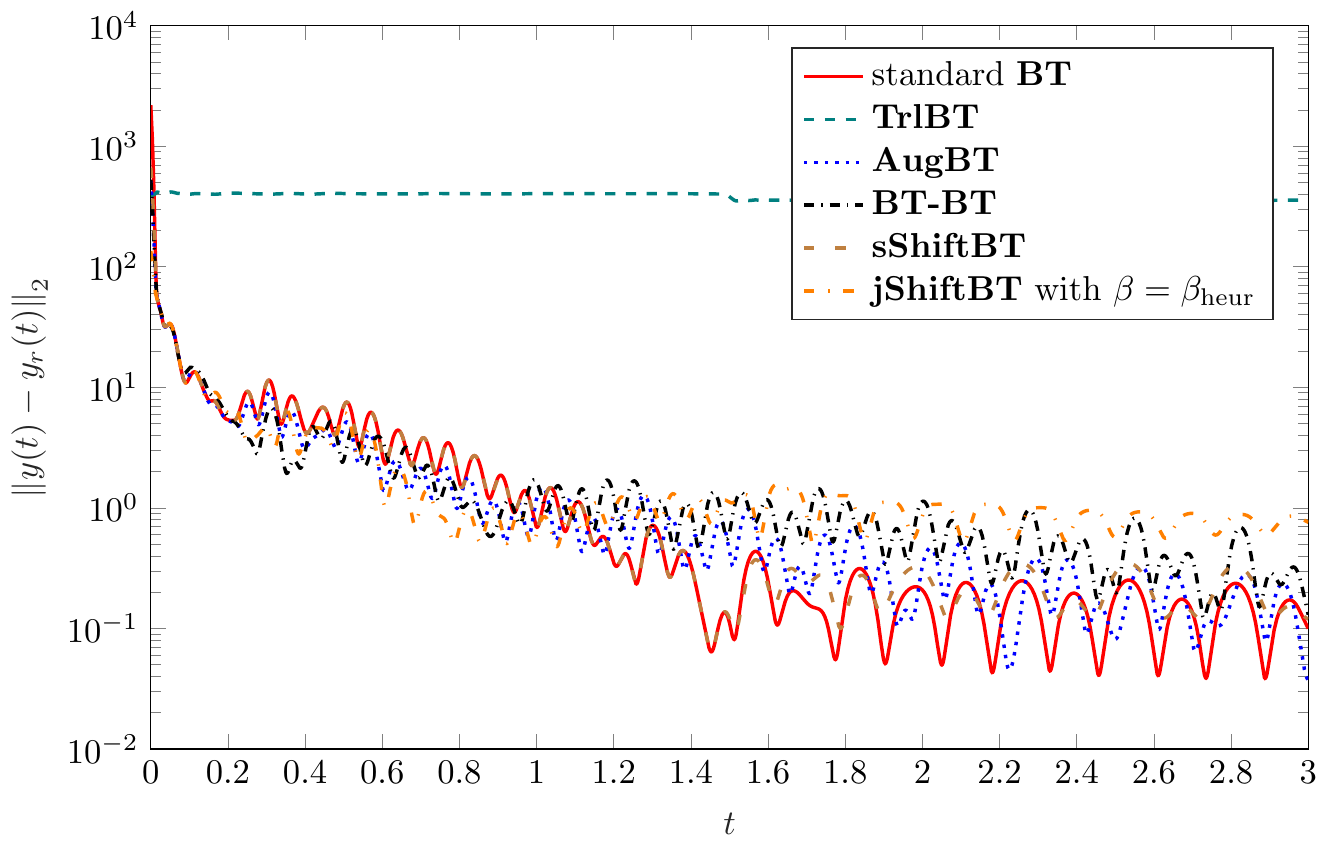}
 } 
 \end{tabular}
 \caption{Comparison of the errors of the new approach with the methods from the literature. Here we use $r=30$ for standard \textbf{BT}, \textbf{TrlBT}, \textbf{AugBT}, and \textbf{jShiftBT}; and $k = \ell = 15$ for \textbf{BT-BT} and \textbf{sShiftBT} as well as the optimized values of $\alpha$.}
 \label{fig:allmethods}
\end{figure} 


Next we show error plots for different choices of $\beta$ in \textbf{jShiftBT} for the reduced order $r=30$ (where we have again applied a smoothing filter). As expected, the figure shows that for smaller values of $\beta$, the reduction of the initial value response is emphazised, leading to smaller errors in the first half of the considered time interval, but to larger errors in the second half. In contrast to that, we focus on the reduction of the input response for larger $\beta$ and so the errors are bigger in the first, but smaller in the second half of the time interval of interest.


\begin{figure}
  \centering
  \begin{tabular}{c}
  \subfigure[Results for Example~\ref{exm:beam}.]{
 \includegraphics{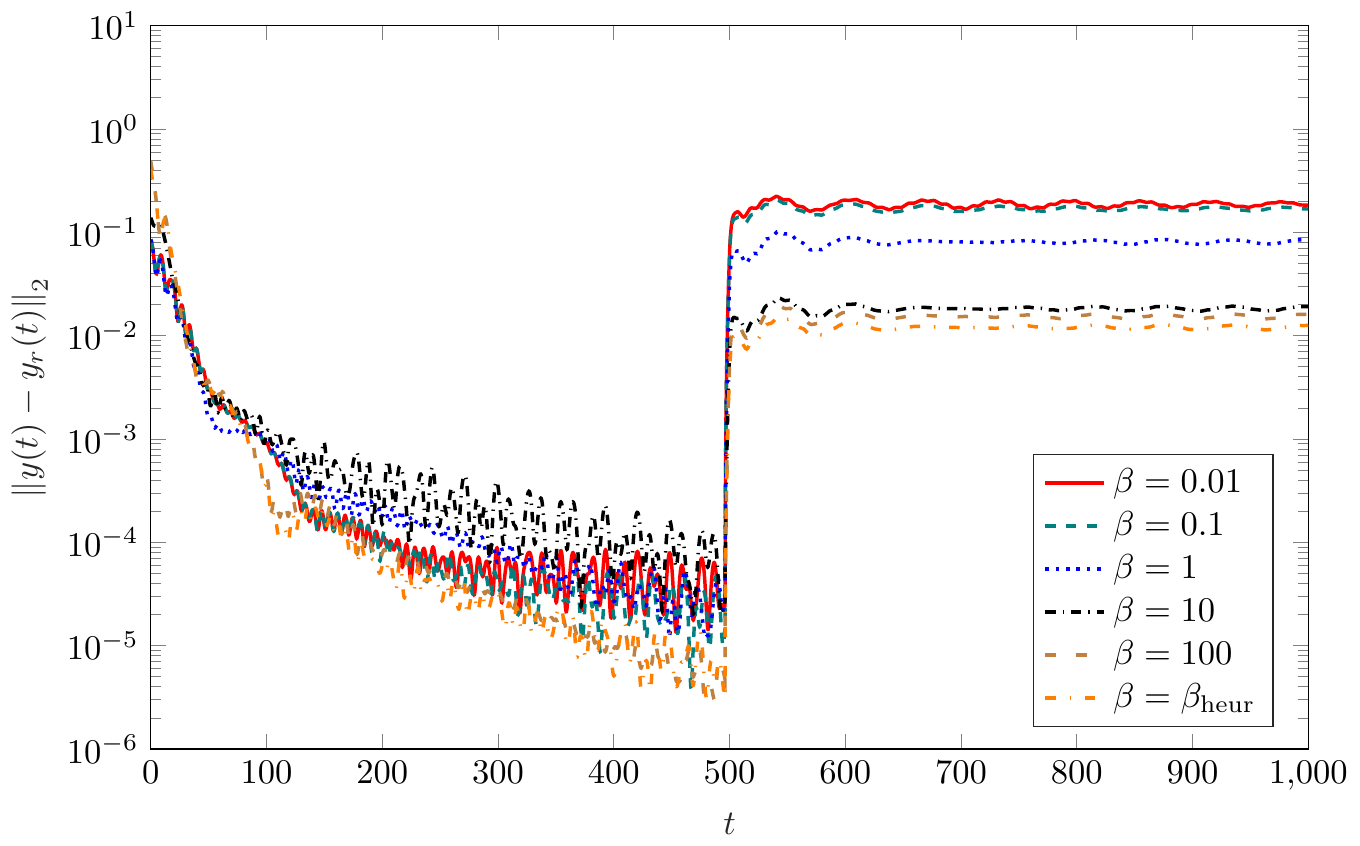}
  } \\
  \subfigure[Results for Example~\ref{exm:CDplayer}.]{
 \includegraphics{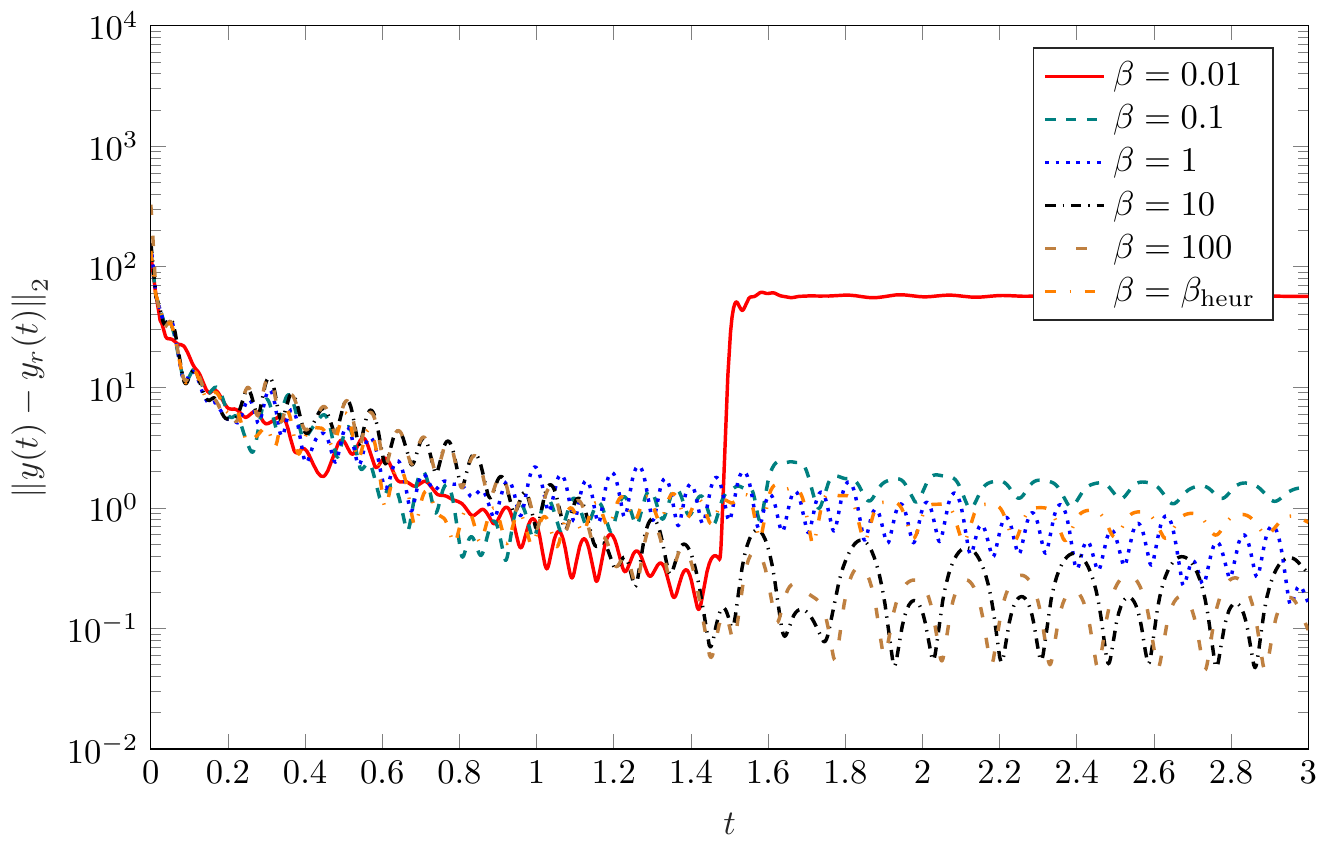}
  } 
  \end{tabular}
 \caption{Comparison of the errors of \textbf{jShiftBT} for various fixed values of $\beta$ and the optimal choice of $\alpha$. Here we use the reduced order $r=30$.}
  \label{fig:joint}
\end{figure} 

Finally, we evaluate the $\LL_2$- and $\LL_\infty$-errors obtained in all our numerical simulations and we list the corresponding error norms in Table~\ref{tab:errors}. This table shows that \textbf{jShiftBT} is the clear winner among all the methods and that the overall errors are also quite robust with respect to changes in the parameter $\beta$ for the considered examples. However, a few methods sometimes get close to the errors of \textbf{jShiftBT}, this is especially the case for \textbf{AugBT}. The table further indicates that the reduction obtained by \textbf{BT} for Example~\ref{exm:CDplayer} is relatively poor compared with the other methods. This is due to the fact that the response to the initial value is not reduced well (as purposely designed in this example).

\begin{table}[t]
 \centering
  \caption{Comparison of the error norms for the simulations in Figures~\ref{fig:allmethods} and~\ref{fig:joint}. Here we list the $\LL_2$- and the $\LL_\infty$-norms of the error trajectories in the respective time interval displayed in the figures. The smallest errors are emphasized by bold font.}
  \label{tab:errors}
 \subtable[Results for Example~\ref{exm:beam}]{
\begin{tabular}{c|cc}
        method & $\LL_2$-error & $\LL_\infty$-error \\
        \hline 
      standard \textbf{BT}  & 1.3\epl 0 & 2.2\epl 0\\
      \textbf{TrlBT}  & 2.1\epl 1 & 8.3\emi 1 \\
      \textbf{AugBT}  & 7.8\emi 1 & 5.3\emi 1 \\
      \textbf{BT-BT}  & 1.6\epl 1 & 6.8\epl 1 \\
      \textbf{sShiftBT}  & 1.6\epl 1 & 3.0\epl 0 \\
      \textbf{jShiftBT}, $\beta = \beta_{\rm heur}$  & 1.1\epl 0 &  1.7\epl 0\\
      \textbf{jShiftBT}, $\beta = 0.01$  & 4.2\epl 0 & \textbf{2.3\emi 1}\\
      \textbf{jShiftBT}, $\beta = 0.1$  & 3.8\epl 0 & {\bf 2.3\emi 1}\\
      \textbf{jShiftBT}, $\beta = 1$  & 1.8\epl 0 & 2.6\emi 1\\
      \textbf{jShiftBT}, $\beta = 10$  & {\bf 6.9\emi 1} & 4.0\emi 1\\
      \textbf{jShiftBT}, $\beta = 100$  & 1.3\epl 0 & 2.0\epl 0 \\
  \end{tabular}} 
   \subtable[Results for Example~\ref{exm:CDplayer}]{
\begin{tabular}{c|cc}
        method & $\LL_2$-error & $\LL_\infty$-error \\
        \hline 
      standard \textbf{BT}  & 2.3\epl 2 & 1.2\epl 4\\
      \textbf{TrlBT}  & 6.6\epl 2 & 7.3\epl 2 \\
      \textbf{AugBT}  & 5.2\epl 1 & 2.3\epl 3 \\
      \textbf{BT-BT}  & 5.5\epl 1 & 2.4\epl 3 \\
      \textbf{sShiftBT}  & 6.4\epl 1 & 1.6\epl 3 \\
      \textbf{jShiftBT}, $\beta = \beta_{\rm heur}$  & {\bf 1.9\epl 1} & {\bf 4.9\epl 2} \\
      \textbf{jShiftBT}, $\beta = 0.01$  & 7.2\epl 1 & 5.2\epl 2 \\
      \textbf{jShiftBT}, $\beta = 0.1$  & {\bf 1.9\epl 1} & {\bf 4.9\epl 2} \\
      \textbf{jShiftBT}, $\beta = 1$  & 2.0\epl 1 & 5.7\epl 2 \\
      \textbf{jShiftBT}, $\beta = 10$  & 2.0\epl 1 & 5.4\epl 2 \\
      \textbf{jShiftBT}, $\beta = 100$  & 3.6\epl 1 & 1.3\epl 3 \\
  \end{tabular}}
\end{table}

\section{Concluding Remarks}
In this work we have derived a new alternative procedure for balanced truncation model reduction for systems with nonzero initial value. In contrast to other methods, our method provides an a priori error bound that can be computed efficiently from the solutions of three Lyapunov equations that are needed in the reduction algorithm. As the numerical examples have shown, our error bound and also the errors are often better to those of other techniques available in the literature. Especially our new joint projection method outperforms all the other methods in our numerical experiments. Even, if multiple joint projection ROMs have to be used for a large range of inputs and initial values, they can be constructed very efficiently. This is because the main computational burden is the solution of the three Lyapunov equations, whereas the parameter optimization is comparably cheap. Therefore, we recommend the potential practitioner to use \textbf{jShiftBT} for reducing models with nonzero initial condition.   

\section*{Code Availability}
The \textsc{Matlab} code and data for reproducing the numerical results are 
available for download under the DOI \texttt{10.5281/zenodo.6355512}.

\section*{Acknowledgement}
The authors thank Bj\"orn Liljegren-Sailer (Universit\"at Trier) for noticing an error in a previous version of our software.

\bibliographystyle{abbrv}
\bibliography{SchV20}

\end{document}